\documentclass[letterpaper, 10 pt, conference]{ieeeconf}  % Comment this line out if you need a4paper
\IEEEoverridecommandlockouts                              % This command is only needed if 
  \let\labelindent\relax
\usepackage{amsmath} % assumes amsmath package installed
\usepackage{amssymb}  % assumes amsmath package installed

\usepackage{tabulary}
\usepackage[english]{babel}
\usepackage{blindtext}
\usepackage[font=footnotesize]{caption}
\usepackage[makeroom]{cancel}
\usepackage{amsfonts}
\usepackage{subfig,pgfplots}
\usepackage{amsthm}  %This is for proofs
\usepackage{amssymb}
\usepackage{graphicx}
\usepackage{pgf,tikz}
\pgfplotsset{ % Here we specify options for all figures in the document
  compat=newest, % Which version of pgfplots do we want to use?
   legend style =
  {font=\footnotesize },
  label style = {font=\footnotesize},
every tick label/.append style={font=\footnotesize}
  }
\usepackage{tabularx}
\usepackage{float}
\usepackage{cite}
\usepgfplotslibrary{fillbetween}
  \usepgfplotslibrary{ternary}
\usetikzlibrary{shapes.multipart}
\usepackage{enumitem} % To reference numbered items
\usepackage{bbm,xurl} % This is for \mathbb for lower case letters

\usepackage[strict]{changepage}

\allowdisplaybreaks

\renewcommand{\eqref}[1]{Eq.~(\ref{#1})}  %Modified equation reference

\newtheorem{remark}{Remark}
\newtheorem{theorem}{Theorem}
\newtheorem{lemma}{Lemma}
\newtheorem{assumption}{Assumption}

\newtheorem{corollary}{Corollary}
\newtheorem{proposition}{Proposition}

\title{\LARGE \bf
A Human-Vector Susceptible--Infected--Susceptible Model for Analyzing and Controlling the Spread of Vector-Borne Diseases}

\author{Lorenzo Zino, Alessandro Casu, and Alessandro Rizzo
\thanks{The authors are with the Department of Electronics and Telecommunications, Politecnico di Torino, Torino, Italy (\texttt{lorenzo.zino@polito.it, alessandro.casu@studenti.polito.it, alessandro.rizzo@polito.it}).}}

\begin{document}

\maketitle
\thispagestyle{empty}

\begin{abstract}
We propose an epidemic model for the spread of vector-borne diseases. The model, which is built extending the classical susceptible--infected--susceptible model, accounts for two populations ---humans and vectors--- and for cross-contagion between the two species, whereby humans become infected upon interaction with carrier vectors, and vectors become carriers after interaction with infected humans. We formulate the model as a system of ordinary differential equations and  leverage monotone systems theory to rigorously characterize the epidemic dynamics. Specifically, we characterize the global asymptotic behavior of the disease, determining conditions for quick eradication of the disease (i.e., for which all trajectories converge to a disease-free equilibrium), or convergence to a (unique) endemic equilibrium. Then, we incorporate two control actions: namely, vector control and incentives to adopt protection measures. Using the derived mathematical tools, we assess the impact of these two control actions and determine the optimal control policy. 
\end{abstract}

\section{Introduction}\label{sec:intro}

In the last decade, mathematical models of epidemic diseases have gained traction within the systems and control community~\cite{Nowzari2016,Mei2017,Pare2020review,zino2021survey,Ye2023competitive}. In fact, the development of increasingly refined models has allowed to accurately predict the course of an epidemic outbreak and, ultimately, to design and assess intervention policies~\cite{Nowzari2016,preciado2018vaccine,zino2021survey}. In particular, the latest epidemiological threats, such as the outbreaks of Ebola, COVID-19, and seasonal flu have provided further motivation to pursue these studies, yielding tailored versions of these general epidemic models~\cite{Rizzo2016Ebola,giulia,Carli2020,parino2021,fiandrino2025}. 

Typical modeling setups deal with human-to-human contagion mechanisms~\cite{Nowzari2016,Mei2017,Pare2020review,zino2021survey,Ye2023competitive}. However, according to the World Health Organization,  more than 17\% of all infectious diseases are vector-borne~\cite{WHOvector}. This means that they are not transmitted through human-to-human interactions, but by arthropod vectors (such as mosquitoes, fleas, or ticks) that can carry pathogens and transmit them to humans~\cite{WHOvector}. Vector-borne diseases (including dengue, malaria, and West Nile fever) pose a significant threat to our society, being causing more than 700,000 deaths annually~\cite{WHOvector}. Moreover, the ongoing climate change crisis exacerbates concerns on the prevention of these diseases, as vectors adapt to new habitats~\cite{Rocklv2020}. This is the case, e.g., of \textit{Aedes aegypti} (responsible for the transmission of several diseases, including dengue, Zika, and yellow fever), which is predicted to infest many regions of Europe if the temperature increases by  2\textdegree C~\cite{LiuHelmersson2019}.

Numerous mathematical models of vector-borne diseases have been proposed and studied, particularly in response to the increasing concern for dengue fever, leading to a rich body of research~\cite{Aguiar2022,Ogunlade2023}. However, most of these models, developed by computational epidemiologists as complex simulation tools, offer limited analytical tractability~\cite{Keeling2011}. Conversely, there is a scarcity of parsimonious models that efficiently balance accuracy and interpretability.

Here, we fill in this gap by proposing a novel mathematical model for vector-borne diseases. Our model, grounded in dynamical systems theory, considers two interacting populations of humans and vectors. Through such interactions, the pathogen is transmitted from carrier vectors to susceptible humans and from infectious humans to vectors, establishing a positive feedback loop of contagion. Formally, we cast our model as a system of nonlinear ordinary differential equations (ODEs), in which we couple i) an epidemic model for humans, inspired by the Susceptible--Infected--Susceptible (SIS) model~\cite{Mei2017}, ii) a contagion model for vectors, inspired by the Susceptible--Infected (SI) model~\cite{Mei2017}, and iii) a vital dynamics for vectors, which is modeled using a birth-death process~\cite{Murray1993}. We refer to the model obtained as the human-vector SIS (HV-SIS) epidemic model.

In addition to the formulation of the model, the main contribution of this paper is twofold. First, by leveraging monotone dynamical systems theory~\cite{Hirsch2006}, we perform a thorough analysis of the asymptotic behavior of the HV-SIS model, characterizing two regimes: one where the epidemic outbreak is quickly eradicated, leading to global convergence to a disease-free equilibrium; and one where the disease becomes endemic, and the system converges to a (unique) endemic equilibrium. Second, we introduce two control actions: namely, vector control ---which focuses on reducing the vector population of vectors (e.g., using pesticides)~\cite{Wilson2020}--- and the use of personal protection measures against contagion~\cite{Alpern2016}. By studying the controlled HV-SIS model and formulating an optimization problem, we investigate the optimal control policies to prevent outbreaks of vector-borne diseases, as a function of the model parameters and the cost associated with implementing interventions.

\section{Human-Vector SIS Epidemic Model}\label{sec:model}

We consider a large population of humans that interact with a population of vectors. Similar to most epidemic models~\cite{Mei2017}, we observe that the duration of an epidemic outbreak is typically negligible with respect to the life-span of humans. Hence, we approximate the size of the human population as constant. Moreover, being the population large, we approximate it as a continuum of individuals with total mass equal to $1$~\cite{Mei2017}. On the contrary, the life-span of a vector is typically comparable with the infection propagation dynamics~\cite{Loureno2010}. Hence, we assume that the total quantity of vectors $v(t)\geq 0$ (normalized with respect to the unit mass human population) evolves in continuous-time $t\geq 0$ according to a classical ODE associated with a birth-death process, typically used in mathematical biological models~\cite{Murray1993}:
\begin{equation}\label{eq:bd}
    \dot v(t)=\omega-\mu v(t),
\end{equation}
where $\omega> 0$ and $\mu >0$ are two constants representing the birth and death rate, respectively.

Humans can be healthy and susceptible to the disease or infected with the disease. We assume that there is no natural immunity: after recovery, individuals are again susceptible to the disease. This is a good proxy for many diseases, e.g., dengue fever, for which natural immunity wanes quickly and it only protects against the virus serotype specific of the previous infection~\cite{WHOdengue}. We denote by $x(t)\in[0,1]$ and $s(t)\in[0,1]$ the fraction of infected individuals and susceptible individuals at time $t\geq 0$, respectively. Since there is no immunity, it holds $s(t)=1-x(t)$. Similarly, vectors can be either carriers of the pathogen or non-carriers. We denote by $y(t)\geq 0$ the number of non-carrier vectors and by $z(t)\geq 0$ the quantity of carrier vectors. Being $v(t)$ the total quantity of vectors at time $t$, then $y(t)+z(t)=v(t)$. 

\begin{figure}
    \centering
\includegraphics[]{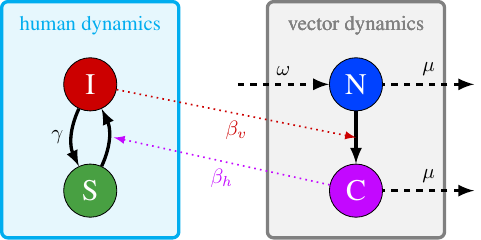}    \caption{Schematic of the human-vector epidemic model. Solid arrows represent possible transitions of the state of humans (S and I for susceptible and infected, respectively) and vectors (N and C for non-carrier and carrier, respectively). Dashed arrows are associated with the vital dynamics of vectors. Dotted colored arrows indicate transitions that are triggered by interactions with humans or vectors with a specific state. }
    \label{fig:schematic}
\end{figure}

We assume that the two populations are well-mixed, and we define a human-vector compartmental model that describes the evolution of the fraction of susceptible and infected individuals and the quantity of carriers and non-carriers in the two populations. The compartmental model, illustrated in Fig.~\ref{fig:schematic}, yields the following 3-dimensional system of nonlinear ODEs:
\begin{subequations}\label{eq:model}
    \begin{align}
        &\dot{x}(t) = -\gamma x(t) + \beta_h (1-x(t)) z(t)\label{eq:model_x}\\
        &\dot{y}(t) = \omega - \mu y(t) - \beta_v x(t) y(t)\label{eq:model_y}\\
        &\dot{z}(t) = \beta_v x(t) y(t) - \mu z(t)\label{eq:model_z},
    \end{align}
\end{subequations}
with initial condition in the domain $\mathcal D:=\{(x,y,z): x,y,z\geq 0, x\leq 1\}$. In the following paragraphs, we extensively discuss these equations.

In \eqref{eq:model_x}, the fraction of infected individuals evolves according to two contrasting mechanisms: the negative contribution $-\gamma x(t)$ accounts for infected individuals who recover at a rate $\gamma> 0$; the positive contribution $\beta_h (1-x(t)) z(t)$ accounts for new infections, whose number is proportional to the quantity of susceptible humans, the quantity of carriers, and a parameter $\beta_h> 0$ that captures the human contagion rate (i.e., the likelihood that the pathogen is transmitted from a carrier vector to a human through a human-vector interaction). This equation resembles the classical SIS
 epidemic model~\cite{Mei2017}, but here new contagions, 
 instead of being proportional to the quantity of infected humans, are proportional to the quantity of carriers. For this reason, we shall refer to the model with dynamics in \eqref{eq:model} and initial condition in $\mathcal D$ as the human-vector SIS  model, abbreviated as \emph{HV-SIS model}.

The other two equations, Eqs.~(\ref{eq:model_y})--(\ref{eq:model_z}), govern the dynamics of vectors. In particular, the term $\beta_v x(t) y(t)$ captures new carriers and gives a positive contribution to the dynamics of carriers and a negative contribution to non-carriers. This term is proportional to the number of non-carrier vectors, infected humans, and a parameter $\beta_v\geq 0$ that captures the vector contagion rate (i.e., the likeliness that the pathogen is transmitted from a human to a vector through a human-vector interaction). The other two terms come from \eqref{eq:bd}: new born vectors are not carriers of the pathogen (so the rate $\omega$ appears in \eqref{eq:model_y}), while the death rate is independent of the pathogen, since vectors are only carriers and not infected with the disease, leading to the terms $-\mu y(t)$ and $-\mu z(t)$, respectively.

\section{Main Results on the HV-SIS Epidemic Model}\label{sec:analysis}

In this section, we present our main results on the analysis of the HV-SIS epidemic model. First, we prove that the equations are well-defined.

\begin{lemma}\label{lemma:invariance}
The domain of the HV-SIS model $\mathcal D:=\{(x,y,z): x,y,z\geq 0, x\leq 1\}$ can be split into two  domains $\mathcal D_1:=\{(x,y,z): x,y,z\geq 0, x\leq 1, y+z\leq \frac{\omega}{\mu}\}$ and $\mathcal D_2:=\{(x,y,z): x,y,z\geq 0, x\leq 1, y+z\geq \frac{\omega}{\mu}\}$, which are positive invariant under~\eqref{eq:model}. 
\end{lemma}
\begin{proof}
The domain $\mathcal D$ is closed and convex and the vector field in \eqref{eq:model} is Lipschitz-continuous. 
Hence, Nagumo’s Theorem can be applied~\cite{blanchini1999set}. 
We need to verify that the vector field at the boundaries of the domain does not point towards the boundary. We immediately observe that, if any of the variables is equal to $0$, then the corresponding derivative is always non-negative (hence, it does not point towards the boundary). Similarly, at $x=1$, we get that \eqref{eq:model_x} is always non-positive. Finally, when $y+z=
\frac{\omega}{\mu}$, from summing \eqref{eq:model_y} and \eqref{eq:model_z} we get $\dot y+\dot z=0$. Hence, the vector field does not point towards any boundary, yielding the first claim. The second claim follows the same arguments, where we observe that it always holds that $\dot y+\dot z\leq 0$ in $\mathcal D_2$. 
\end{proof}

Then, we provide a complete characterization of the asymptotic behavior of the HV-SIS model, determining its equilibria. Specifically, we will prove that, depending on the model parameters, there is always one equilibrium that is (almost) globally asymptotically stable, characterizing two distinct regimes: either the disease is eradicated and all trajectories converge to a disease-free equilibrium (DFE), or the disease becomes endemic and (almost) all trajectories converge to an endemic equilibrium (EE), where a fraction of the population is infected (and a fraction of the vectors are carriers). The phase transition between these two  regimes, which is a typical phenomenon of many epidemic models~\cite{Mei2017,zino2021survey}, is shaped by the value of the model parameters that determine the so-called \emph{epidemic threshold}~\cite{zino2021survey}. We start our analysis by determining the equilibria of \eqref{eq:model} and determining their (local) stability.

\begin{proposition}\label{prop:equilibria}
The HV-SIS model in \eqref{eq:model} has, at most two equilibria: i) The DFE
    \begin{equation}\label{eq:dfe}
       (x^*,y^*,z^*)=\Big(0, \frac{\omega}{\mu}, 0\Big),
    \end{equation}
and ii) the EE
    \begin{equation}\label{eq:ee}
               (\bar x,\bar y,\bar z)=\Big(\frac{\omega \beta_h \beta_v - \mu^2 \gamma}{\omega \beta_h \beta_v + \mu \gamma \beta_v} , \frac{\gamma \mu + \beta_h \omega}{\beta_h (\beta_v + \mu)}, \frac{\omega \beta_h \beta_v - \mu^2 \gamma}{\mu \beta_h \beta_v + \mu^2 \beta_h}\Big).
    \end{equation}
 Specifically, let us define the epidemic threshold 
    \begin{equation}\label{eq:threshold}
\sigma_0:=\frac{\beta_h \beta_v \omega}{\gamma \mu^2}.
\end{equation}
The DFE in \eqref{eq:dfe} always exists and is locally exponentially stable if $\sigma_0< 1$ and unstable  if $\sigma_0> 1$. The EE in \eqref{eq:ee} exists and is distinct from the DFE if and only if $\sigma_0>1$ and (if it exists) it is always locally exponentially stable.
\end{proposition}
\begin{proof}
First, we compute the equilibria of \eqref{eq:model} by equating the right hand sides to $0$,  obtaining a system of three nonlinear equations, which yields the two solutions in \eqref{eq:dfe} and \eqref{eq:ee}. Then, we observe that the DFE is always in the domain $\mathcal D$. On the contrary, the EE is in the domain $\mathcal D$ if and only if the numerators of $\bar x$ and $\bar z$ are non-negative, i.e., if $\omega \beta_h \beta_v - \mu^2 \gamma\geq 0$, which yield the condition $\frac{\omega \beta_h \beta_v}{\mu^2 \gamma}\geq 1$. Finally, we observe that, when $\frac{\omega \beta_h \beta_v}{\mu^2 \gamma}= 1$, the DFE and the EE coincide, yielding the strict inequality for the existence of a second equilibrium of \eqref{eq:model}.

At this stage, we compute the Jacobian matrix of \eqref{eq:model} in a generic point $(x,y,z)$, that is, \begin{equation}\label{eq:jacobian}
    J(x,y,z) = \begin{bmatrix}
    -\gamma - \beta_h x & 0 & \beta_h (1-x) \\
    -\beta_v y & -\mu - \beta_v x & 0 \\
    \beta_v y & \beta_v x & -\mu
\end{bmatrix}.
\end{equation}
By evaluating \eqref{eq:jacobian} at the DFE in \eqref{eq:dfe}, we get
 \begin{equation}
        J( x^*,y^*, z^*) = \begin{bmatrix}
    -\gamma & 0 & \beta_h  \\
    -\beta_v \frac{\omega}{\mu} & -\mu  & 0 \\
    \beta_v \frac{\omega}{\mu} & 0 & -\mu\end{bmatrix}, 
 \end{equation}
whose eigenvalues are $\lambda_1 = -\mu$ and $\lambda_{2} = \frac{1}{2\mu}(-\gamma \mu - \mu^2 - \sqrt{\mu}\sqrt{\gamma^2 \mu - 2\gamma \mu^2 + \mu^3 + 4\beta_h \beta_v \omega})$, whose real parts are always negative, and  $\lambda_{3} = \frac{1}{2\mu}(-\gamma \mu - \mu^2 + \sqrt{\gamma^2 \mu^2 - 2\gamma\mu^3 + \mu^4 + 4\beta_h \beta_v \omega\mu})$, which is negative if and only if $\gamma^2 \mu^2 - 2\gamma\mu^3 + \mu^4 + 4\beta_h \beta_v \omega\mu<(\gamma \mu + \mu^2)^2  $, which simplifies to the condition $\sigma_0=\frac{\beta_v \beta_h \omega}{\gamma \mu^2} < 1$. Hence the DFE is locally exponentially stable if $\sigma_0 < 1$ and unstable if  $\sigma_0 > 1$.

Similarly, we evaluate the Jacobian matrix in \eqref{eq:jacobian} at the EE in \eqref{eq:ee}, obtaining
\begin{equation}
\begin{bmatrix}
    -\frac{\gamma \mu \beta_v + \omega \beta_v \beta_h}{\mu (\mu + \beta_v)} & 0 & \frac{\mu \gamma \beta_v \beta_h + \mu^2 \gamma \beta_h}{\omega \beta_v \beta_h + \mu \gamma \beta_v} \\
    -\frac{\beta_v \gamma \mu + \beta_v \beta_h \omega}{\beta_h (\beta_v + \mu)} & -\mu - \frac{\omega \beta_v \beta_h - \mu^2 \gamma}{\omega \beta_h + \mu \gamma} & 0 \\
    \frac{\beta_v \gamma \mu + \beta_v \beta_h \omega}{\beta_h (\beta_v + \mu)} & \frac{\omega \beta_v \beta_h - \mu^2 \gamma}{\omega \beta_h + \mu \gamma} & -\mu
\end{bmatrix},
\end{equation}
whose eigenvalues are $\lambda_1=-\mu$, which is always negative, and another pair of eigenvalues, which are not reported due to their cumbersome expression. Again, by imposing that the largest of the two has negative real part, we obtain a complicated condition which can be simplified to $\sigma_0>1$, where computations are omitted due to space constraints.
\end{proof}

\begin{remark}\label{rem:threshold}
    From the expression of the epidemic threshold in \eqref{eq:threshold}, we observe that, as predictable, increasing the infection rates $\beta_h$ and $\beta_v$ favors the spread of the disease. A similar effect is observed by  increasing the vector birth rate $\omega$. On the other hand, increasing the vector death rate $\mu$ and/or the human recovery rate $\gamma$ favors the eradication of the disease. Interestingly, the vector death rate $\mu$ has a larger impact, since it appears squared at the denominator, suggesting that vector control is a potentially effective strategy to avoid outbreaks of vector-borne diseases. 
\end{remark}

Proposition~\ref{prop:equilibria} characterizes the local behavior of \eqref{eq:model} about the two equilibria of the system. In order to prove global convergence, we now leverage monotone systems theory~\cite{Hirsch2006}. However, since the Jacobian of \eqref{eq:model} in \eqref{eq:jacobian} is evidently not a Metzler matrix, we cannot  directly apply the  monotone systems theory to \eqref{eq:model}, and we need to introduce a change of variables, as detailed in the proof of the following result.

\begin{theorem}\label{th:uncontrolled}
Let $\sigma_0$ be the epidemic threshold from \eqref{eq:threshold}. If $\sigma_0\leq 1$, then all trajectories of the HV-SIS model in \eqref{eq:model} converge to the DFE in \eqref{eq:dfe}. If $\sigma_0> 1$, then all trajectories with initial condition such that $x(0)\neq 0$ or $z(0)\neq 0$ converge to the EE in \eqref{eq:ee}.
\end{theorem}
\begin{proof}
    We operate a change of variables, where we introduce an auxiliary 3-dimensional system formed by $x(t)$, $z(t)$, and $v(t)=y(t)+z(t)$, which is governed by \eqref{eq:bd}. Being $y(t)=v(t)-z(t)$, we obtain 
    \begin{subequations}\label{eq:auxiliary}
    \begin{align}
        &\dot{x}(t) = -\gamma x(t) + \beta_h (1-x(t)) z(t)\label{eq:auxiliary_x}\\
        &\dot{z}(t) = \beta_v x(t) (v(t)-z(t)) - \mu z(t)\label{eq:auxiliary_z}\\
        &\dot{v}(t) = \omega-\mu v(t)\label{eq:auziliary_v}.
    \end{align}
\end{subequations}
    First, from Lemma~\ref{lemma:invariance}, we derive that the two invariant sets, written in terms of the new variables are $\mathcal D_1:=\{(x,z,v): x,z,v\geq 0, x\leq 1, v\leq \frac{\omega}{\mu}, z\leq v\}$ and $\mathcal D_2:=\{(x,z,v): x,z,v\geq 0, x\leq 1, v\geq \frac{\omega}{\mu}, z\leq v\}$. Second, we prove that all trajectories in $\mathcal D_2$ are bounded (those in $\mathcal D_1$ are necessarily bounded, being $\mathcal D_1$ compact). From \eqref{eq:auziliary_v}, we observe that $\dot{v}(t)\leq0$ for any $v(t)\geq\frac{\omega}{\mu}$, which implies $v(t)\leq\max\{v(0),\omega/\mu\}$, which in turn implies that  trajectories cannot diverge. Third, we compute the Jacobian matrix of \eqref{eq:auxiliary} at the generic point $(x,z,v)$, obtaining
    \begin{equation}\label{eq:auxiliary_jacobian}
        \tilde J(x,z,v)=\begin{bmatrix}
    -\gamma - \beta_hz & \beta_h (1-x) &  0\\
    \beta_v(v-z) & -\beta_vx-\mu & \beta_vx \\
0 & 0 & -\mu\end{bmatrix}.
    \end{equation}    
    We observe that the matrix in \eqref{eq:auxiliary_jacobian} is Metzler, since all its off-diagonal entries  are non-negative for values of the parameters belonging to the two invariant sets. Hence, the dynamical system in \eqref{eq:auxiliary} is monotone~\cite{Hirsch2006}, which implies that all its trajectories converge to a fixed point~\cite{Hirsch2006}. Fourth, since $v(t)$ and $z(t)$ converge, also their difference $y(t)=z(t)-x(t)$ necessarily converges to a fixed point, yielding that also trajectories of \eqref{eq:model} converge. Fifth, the analysis of the local stability of the equilibria in Proposition~\ref{prop:equilibria} yields the claim for $\sigma_0<1$ and $\sigma_0>1$. Finally, we observe that when $\sigma_0=1$, Proposition~\ref{prop:equilibria} does not provide any information on the stability of the equilibria, but it states that the system has a unique equilibrium: the DFE. Combining this with the system's monotonicity (which implies convergence to an equilibrium), we obtain convergence to the DFE also in the case $\sigma_0=1$, yielding the claim.
\end{proof}

Figure~\ref{fig:trejectories} illustrates the results of Theorem~\ref{th:uncontrolled}. If $\sigma_0<1$, then the disease is quickly eradicated and the system converges to the DFE. It is interesting to notice from Fig.~\ref{fig:below} that, unlike classical SIS-models~\cite{Mei2017}, the convergence to the DFE may be non-monotone, with an initial increase in the epidemic prevalence $x(t)$, until a peak is reached and the fraction of infected individuals starts decreasing to $0$. On the contrary, if we increase $\sigma_0$ to reach a value larger than $1$ (e.g., by decreasing $\mu$ as in Fig.~\ref{fig:above}), we enter the endemic regime and trajectories converge to the EE.

\begin{figure}
    \centering
   \subfloat[$\mu=0.2$]{\includegraphics{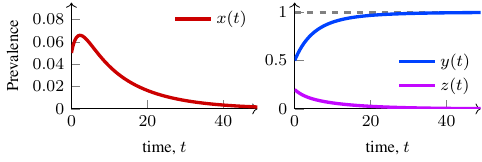}\label{fig:below}}\\
   \subfloat[$\mu=0.1$]{\includegraphics{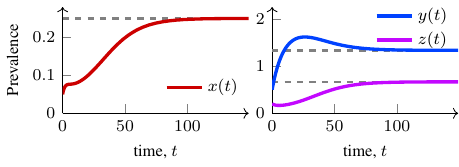}\label{fig:above}}
\caption{Trajectories of the human-vector epidemic model. In (a) $\sigma_0<1$ and the trajectory converges to the DFE; in (b) $\sigma_0>1$ and the trajectory converges to the EE. The equilibrium predicted by Theorem~\ref{th:uncontrolled} in the two cases is depicted with gray dashed horizontal lines. Common parameters are $\omega=\beta_h=\beta_v=0.2$, and $\gamma=0.4$. }
    \label{fig:trejectories}
\end{figure}

\section{Control of the HV-SIS Model}\label{sec:control}

In this section, we consider two distinct control actions that can be implemented in the prevention of vector-borne diseases, and we encapsulate them within the dynamical system in \eqref{eq:model} by introducing two additional terms that capture these actions:
\begin{itemize}
    \item {\bf Vector control} interventions, which focuses on reducing the vector population (e.g., using insecticide-based tools or integrated pest management)~\cite{Wilson2020}; 
    \item Incentives to adopt personal {\bf protection measures}, such as promoting the use of insect repellent,  wear long clothing, and limit outdoor activity~\cite{Alpern2016},
\end{itemize}
Both these interventions have been proven effective in preventing vector-borne diseases~\cite{Wilson2020,Alpern2016}. In the following, we use the mathematical model developed in Section~\ref{sec:model} to analytically assess the effectiveness of these control actions and use them to design an optimal control strategy.

In order to incorporate vector control in the model, we introduce a control parameter $u_1\geq 0$ that captures the efficacy of this control action. In particular, we assume that the death rate of the vector is increased by a parameter $u_1$ thanks to the impact of vector control actions. Then, we introduce a control parameter $u_2\in[0,\beta_h]$ that represents the efficacy of personal protection measures in reducing human contagion rate. Hence, the controlled HV-SIS model is captured by the following system of ODEs:
\begin{subequations}\label{eq:model_controlled}
    \begin{align}
        &\dot{x}(t) = -\gamma x(t) + (\beta_h-u_2) (1-x(t)) z(t)\label{eq:model_controlled_x}\\
        &\dot{y}(t) = \omega - (\mu+u_1) y(t) - \beta_v x(t) y(t)\label{eq:model_controlled_y}\\
        &\dot{z}(t) = \beta_v x(t) y(t) - (\mu+u_1) z(t)\label{eq:model_controlled_z}.
    \end{align}
\end{subequations}

Repeating the same analysis performed for the uncontrolled model in Section~\ref{sec:analysis}, we obtain the following result, whose proof is a corollary of Proposition~\ref{prop:equilibria} and Theorem~\ref{th:uncontrolled}.

\begin{corollary}\label{cor:controlled}
Let 
 \begin{equation}\label{eq:threshold_controlled}
\sigma_c:=\frac{(\beta_h-u_2) \beta_v \omega}{\gamma (\mu+u_1)^2}.
\end{equation}
If $\sigma_c\leq 1$, then all trajectories of the controlled HV-SIS model in \eqref{eq:model_controlled} converge to the DFE 
    \begin{equation}\label{eq:dfe_controlled}
       (x^*_c,y^*_c,z^*_c)=\Big(0, \frac{\omega}{\mu+u_1}, 0\Big).
    \end{equation}
If $\sigma_c> 1$, then all trajectories of the controlled HV-SIS model in \eqref{eq:model_controlled} with initial condition such that $x(0)\neq 0$ or $z(0)\neq 0$ converge to the EE 
     \begin{subequations}\label{eq:ee_controlled}\begin{align}
        &\bar x_c = \frac{\omega  \beta_v (\beta_h-u_2)- (\mu+u_1)^2 \gamma}{\omega \beta_v (\beta_h-u_2)+ (\mu+u_1) \gamma \beta_v}\\
        &\bar y_c= \frac{\gamma (\mu+u_1) + (\beta_h-u_2) \omega}{(\beta_h-u_2) (\beta_v + \mu+u_1)}\\
        &\bar z_c= \frac{\omega  \beta_v (\beta_h-u_2)- (\mu+u_1)^2 \gamma}{(\mu+u_1) (\beta_h-u_2) \beta_v + (\mu+u_1)^2 (\beta_h-u_2)}.
    \end{align}\end{subequations}
\end{corollary}

\begin{figure}
    \centering
\subfloat[Infected individuals, $\bar x_c$]{\includegraphics[]{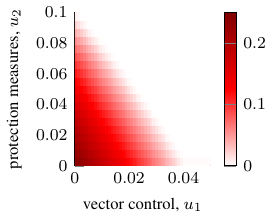}\label{fig:control_x}}\,\,\subfloat[Carrier vectors, $\bar z_c$]{\includegraphics[]{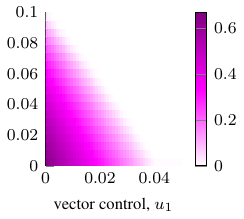}\label{fig:control_z}}    
\caption{(a) Infected individuals and (b) carrier vectors at the EE for different values of the control inputs $u_1$ and $u_2$. Model parameters are $\omega=\beta_h=\beta_v=0.2$, $\gamma=0.4$, and $\mu=0.1$.}
    \label{fig:control}
\end{figure}

The expressions derived in Corollary~\ref{cor:controlled} can be used to assess the performance of the two different control strategies. In Fig.~\ref{fig:control}, we report the value of the fraction of infected individuals at the EE in \eqref{eq:ee_controlled} for different values of the control inputs $u_1$ and $u_2$. The figure suggests that vector control is more effective not only in reducing the epidemic threshold (as observed in Remark~\ref{rem:threshold}), but also in reducing the fraction of infected individuals at the EE. In fact, from Fig.~\ref{fig:control_x}, we observe that when $u_1=0.01$, which means an increase in the vector death rate by just 10\%, the prevalence at the EE decreases by more than 25\%. To obtain the same results using only protective measures, one needs to have $u_2$ that reduces the human infection rate by almost $20\%$. The same relation is observed for the vectors in Fig.~\ref{fig:control_z}.

From Corollary~\ref{cor:controlled} and the following discussion, a problem spontaneously arises. How can one design an optimal control policy in terms of vector control and protection measures to achieve eradication of the outbreak, minimizing the cost of the control strategy? Formally, we can define a cost function $C(u_1,u_2)$, associated with implementing level $u_1$ and $u_2$ of  vector control and protection measures, respectively, for which it is reasonable to make the following assumptions.
\begin{assumption}\label{a:cost}
    The cost function $C(u_1,u_2):[0,\infty)\times[0,\beta_h]\to[0,\infty)$ is a non-negative differentiable function and it is monotonically increasing in $u_1$ and $u_2$.
\end{assumption}

Then, we  formulate the following optimization problem:
\begin{equation}\label{eq:problem}
    \begin{array}{rl}
    (u_1^*,u_2^*)=\arg\min & C(u_1,u_2)\\
       \text{subject to} & (\beta_h-u_2) \beta_v \omega-\gamma (\mu+u_1)^2\leq 0\\
       & u_1,u_2\geq 0,
       \\&u_2\leq \beta_h,
    \end{array}
\end{equation}
where constraint $(\beta_h-u_2) \beta_v \omega-\gamma (\mu+u_1)^2\leq 0$ is obtained from \eqref{eq:threshold_controlled}, by imposing that the DFE is globally asymptotically stable, i.e., imposing $\sigma_c\leq 1$. From the analysis of the optimization problem in \eqref{eq:problem}, we obtain the following result, which provides an explicit way to compute the optimal control policy for the HV-SIS epidemic model.

\begin{theorem}\label{th:control}
Under Assumption~\ref{a:cost}, the optimal solution $(u_1^*,u_2^*)$ of \eqref{eq:problem} solves the following system of equations:
\begin{equation}\label{eq:optimal_general}
    \begin{array}{l}
         \frac{\partial}{\partial u_1}C(u_1,u_2)-2\lambda\gamma(\mu+u_1)=0\\[3pt]
         \frac{\partial}{\partial u_2}C(u_1,u_2)-\lambda\beta_v\omega=0\\[3pt]
         (\beta_h-u_2) \beta_v \omega-\gamma (\mu+u_1)^2=0.
    \end{array}
\end{equation}
\end{theorem}
\begin{proof}
    First, we observe that the problem is always feasible. In fact, $u_1=0$ and $u_2=\beta_h$ is a solution that satisfies all the constraints. Then, we prove that the minimum of \eqref{eq:problem} is attained for values of the control inputs $u_1$ and $u_2$ that either are both equal to $0$, or they satisfy the equality constraint $(\beta_h-u_2) \beta_v \omega-\gamma (\mu+u_1)^2=0$. To prove this statement, let us define $g(u_1,u_2)=(\beta_h-u_2) \beta_v \omega-\gamma (\mu+u_1)^2$. If $g(0,0)\leq 0$ (which is equivalent to $\sigma_0\leq 1$), then the monotonicity of $C$ implies that the minimum is  attained at $u_1^*=u_2^*=0$. If $g(0,0)>0$, assume that $(u_1^*,u_2^*)$ is the optimal solution of \eqref{eq:problem} and that $g(u_1^*,u_2^*)< 0$. The cost function at the optimal solution is equal to $C(u_1^*,u_2^*)$. If $u_1^*>0$, we can define $\tilde u_1^*(\zeta)=u_1^*-\zeta$. By continuity, being $g(u_1^*,u_2^*)<0$ there exists $\Delta u>0$ such that $g(\tilde u_1^*(\Delta u),u^*_2)\leq0$. Clearly $(\tilde u_1^*(\Delta u),u^*_2)$ is then a feasible solution of \eqref{eq:problem}, and $C(\tilde u_1^*(\Delta u),u^*_2)<C(u_1,u^*_2)$, due to the monotonicity of the cost function, which contradicts the assumption that $(u_1^*,u_2^*)$ is the optimal solution of \eqref{eq:problem}, yielding the claim.  If $u_1^*=0$, the same argument holds letting $\tilde u_2^*(\zeta)=u_2^*-\zeta$.

Once we know that the optimal solution is attained at the boundary $g(u_1,u_2)=0$, we use Lagrange multipliers to solve the optimization problem~\cite{Bertsekas1995programming}, by writing the Lagrangian function
\begin{equation}\label{eq:lagrangian}
\begin{array}{l}
    \mathcal L(u_1,u_2,\lambda)=C(u_1,u_2)+\lambda g(u_1,u_2)
\\
=C(u_1,u_2)+\lambda((\beta_h-u_2) \beta_v \omega-\gamma (\mu+u_1)^2).
    \end{array}
\end{equation}
Finally, a necessary condition for optimality is that the solution should solve the nonlinear system obtained by posing the partial derivatives of \eqref{eq:lagrangian} to $0$~\cite{Bertsekas1995programming}, yielding \eqref{eq:optimal_general}
\end{proof}
In general, \eqref{eq:optimal_general} can have multiple solutions, and thus being a solution of \eqref{eq:optimal_general} is only a necessary condition for optimality. However, in the special case in which the cost function is linear, we can derive a close-form expression for the optimal solution of the control problem in \eqref{eq:problem}.
\begin{corollary}\label{cor:linear}
    Assume that $C(u_1,u_2)=c_1u_1+c_2u_2$, where $c_1> 0$ and $c_2>$ are positive constants that weight the cost for implementing vector control and personal protection measures, respectively. Then, the optimal solution $(u_1^*,u_2^*)$ of \eqref{eq:problem} is given by
\begin{subequations}\label{eq:optimal}
\begin{align}
    &u_1^*=\left\{\begin{array}{ll}\sqrt{\frac{\beta_h\beta_v\omega}{\gamma}}-\mu&\text{if }\sigma_0>1\text{ and }(c_1,c_2)\notin\mathcal C\\
    0&\text{otherwise,}
    \end{array}\right.\label{eq:optimal_1}\\
     &u_2^*=\left\{\begin{array}{ll}\frac{\beta_h\beta_v\omega-\gamma\mu^2}{\beta_v\omega}&\text{if }\sigma_0>1\text{ and }(c_1,c_2)\in\mathcal C\\
    0&\text{otherwise,}
    \end{array}\right.\label{eq:optimal_2}
    \end{align}
\end{subequations}
where
\begin{equation}\label{eq:set_control}
    \mathcal C:=\left\{(c_1,c_2):\frac{c_1}{c_2}>\frac{2\gamma\mu}{\beta_v\omega},
    \Big(\frac{c_1}{c_2}\beta_v\omega-\gamma\mu\Big)^2\geq \beta_h\beta_v\gamma\omega\right\}.
\end{equation} 
\end{corollary}
\begin{proof}
    \eqref{eq:optimal} is obtained as the unique solution of \eqref{eq:optimal_general} for $C(u_1,u_2)=c_1u_1+c_2u_2$. 
\end{proof}

\begin{figure}
    \centering
    \subfloat[$\beta_h=0.4$, $\beta_v=0.1$]{\includegraphics[]{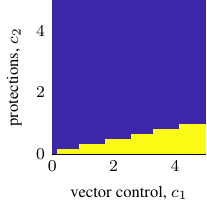}}\quad
    \subfloat[$\beta_h=0.1$, $\beta_v=0.4$]{\includegraphics[]{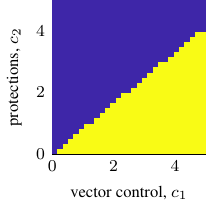}}
\caption{In the blue area, vector control is preferable; in the yellow area, personal protection measures are preferable, according to Corollary~\ref{cor:linear} Common model parameters are $\omega=0.2$, $\gamma=0.4$, and $\mu=0.1$.}
    \label{fig:optimal}
\end{figure}

The results in Corollary~\ref{cor:linear} suggest that, if the cost for implementing intervention policies grows linearly in the effectiveness of the intervention, then it is always beneficial to focus on implementing only one of the two types of policies. Which policy to implement depends on the model parameters and on the ratio between the cost for implementing the two control actions, as illustrated in Fig.~\ref{fig:optimal}. Note that, as the vector infection rate $\beta_v$ becomes larger, the region in which incentivizing protection measures is preferable becomes larger. However, even when $\beta_v$ is four times larger than $\beta_h$, the region in which vector control is more effective is larger.

\section{Conclusion}\label{sec:conclusion}

In this paper, we have proposed and analyzed a novel model for the spread of vector-borne diseases. The model, built using a system of ODEs, accounts for human and  vector contagion, as well as  for the vital dynamics of the vectors. Using systems theoretic tools, we have studied the asymptotic behavior of the system, characterizing a phase transition between a regime where the DFE is globally asymptotically stable and a regime where the system converges to a (unique) EE. Then, by introducing two control actions in the equations, we have analytically assessed the effectiveness of vector control interventions and incentives for humans to adopt protection measures. 

These preliminary results pave the way for several lines of future research. First, the control strategies proposed in Section~\ref{sec:control} are assumed constant. Future research should focus on designing dynamical control strategies, following the approaches developed in~\cite{Nowzari2016,zino2021survey}. Moreover, despite Theorem~\ref{th:control} applies to quite general cost functions, we have focused our discussion on the linear scenario, for which a closed-form expression for the optimal control can be easily derived. Nonlinear cost functions that accounts, e.g., for diminishing returns in the effectiveness of interventions should be investigated. Second, embedding the model on a network structure is a key future step to gain insights into the impact of geographical displacement of humans and vectors on vector-borne diseases. Finally, the HV-SIS model can be coupled with more realistic models of human behavior, e.g., using game theory~\cite{ye2021game,
frieswijk2022mean,hota2022,Paarporn2023},  to develop a more realistic framework to study interventions.
    

\end{document}